\documentclass[11pt,a4paper,english]{article}
\title{ Self-Adjointness of 
 Deformed Unbounded Operators} 
\author{A. Much\\ \footnotesize{Instituto de Ciencias Nucleares, UNAM, M\'exico D.F. 04510, M\'exico}}

\usepackage[pdftex]{graphicx} 
\usepackage{float}            
\usepackage[T1]{fontenc}   
\usepackage[cmex10]{amsmath}  
\usepackage{amstext}          
\usepackage{mathrsfs}
\usepackage{amsfonts}         
\usepackage{amssymb}          
\usepackage{setspace}
\usepackage{bm}               
\usepackage{enumerate}        
\usepackage[bottom]{footmisc} 
\usepackage{array}            
\usepackage{textcomp}         
\usepackage{pdfpages}         
\usepackage{parskip}          
\usepackage[right]{eurosym}   
\usepackage{xcolor}           
\usepackage[square,numbers]{natbib}	  
 
\usepackage[hyphens]{url}     
\usepackage{multicol}         
\usepackage{graphicx}
\usepackage{subfigure}
\usepackage{dsfont}
\usepackage{amsmath}
\usepackage{amssymb}
\usepackage{amsthm}
\usepackage{moreverb}
\usepackage[utf8]{inputenc}
\usepackage{tabularx}
\usepackage{longtable}
\usepackage{nicefrac}
\usepackage[pdf]{pstricks}
\usepackage{tikz}
\usepackage[off]{auto-pst-pdf}
\usepackage[colorlinks,pdfpagelabels,pdfstartview = FitH,bookmarksopen = true,bookmarksnumbered =
true,linkcolor = black,plainpages = false,hypertexnames = false,citecolor = black]{hyperref}

\newtheorem{theorem}{\textsc{Theorem}}[section]
\newtheorem{lemma}{\textsc{Lemma}}[section]
\newtheorem{proposition}{\textsc{Proposition}}[section]

\theoremstyle{definition}
\newtheorem{definition}{\textsc{Definition}}[section]

\theoremstyle{remark}
\newtheorem{remark}{Remark}[section]

\usepackage{fancyhdr}
 \numberwithin{equation}{section} 
\begin{document}
\maketitle
\abstract{
\begin{center}
Dedicated to the 80th Birthday of Prof. Marcos Rosenbaum
\end{center} $\,$\\ We consider deformations of unbounded operators  by using the novel construction tool of warped convolutions.  By using the Kato-Rellich theorem we show that  unbounded self-adjoint deformed operators are self-adjoint if they satisfy a certain condition. This condition proves itself to be necessary for the  oscillatory integral to be well-defined. 
Moreover, different proofs are given for self-adjointness of deformed unbounded operators in the context of quantum mechanics and quantum field theory.} 
  \tableofcontents

\section{Introduction}  
The theory of deformation has an interesting and long history. Initially it was thought of as the so called mathematically rigorous path to the quantization of classical, i.e. Newtonian, observables in quantum mechanics (QM).    However, only recently a rigorous apparatus was constructed in which quantized fields were deformed in order to describe quantum fields living on a quantum space-time. This major achievement was made popular under the name of warped convolutions,  \cite{GL1, BS, BLS}.  Usually, this method
is used in the realm of quantum field theory to deform free quantum fields and to construct
non-trivial interacting fields which was done in \cite{A, GL1, GL2, GL4, GL5, Mor, MUc}. It was also used in quantum measurement theory and quantum mechanics \cite{AA, Muc1}. One of the major advantages of this method   is its easy accessibility to a physical regimen. \newline\newline   By using this novel tool in a quantum mechanical context, we recast  many fundamental physical effects involving electromagnetism.   Moreover, we were able to produce gravitomagnetic effects and interaction between magnetic and gravitomagnetic fields by this  deformation procedure.
 \newline\newline
The idea developed in the QM-case was extended in  \cite{Muc3} to a relativistic quantum field theoretical context.  We first defined a quantum field theoretical version of an operator that is conjugate to the second-quantized momentum. This is done by taking the QM-definition using the momentum operator and performing a second-quantization onto the bosonic Fock space.  Next we deformed the conjugate spatial operator and obtained by calculating the  commutator of the deformed coordinate operators  terms resembling relativistic corrections to the standard Moyal-Weyl. 
 \newline\newline
One question still remains unanswered in the context of QM and QFT operators namely, are the deformed unbounded self-adjoint operators self-adjoint? From the original work we know  \cite{BS, BLS} they are at least symmetric if the undeformed operators are self-adjoint. It is a fact of life that in  quantum physics we have to deal with unbounded operators; hence a thorough investigation of self-adjointness of deformed  operators is given in this work. The author will give two different proofs concerning the self-adjointness of the deformed Hamiltonian and one proof in the QFT-case for the deformation of the Newton-Wigner operator. This is essential since in the physical world an observable is obliged to have real eigenvalues which is only guaranteed when self-adjointness is given. 
\newline\newline
The question of self-adjointness for a general unbounded operator is further investigated. By imposing certain constraints on the action of the automorphisms of the group $\mathbb{R}^n$ on the self-adjoint undeformed operator, essential self-adjointness follows by W\"ust's theorem. 
\newline\newline
This work is organized as follows: In the second  Section 2 we review the important preliminaries for this work, i.e. we give a short but detailed summary of  warped convolutions and the Kato-Rellich theorem. Deformations in quantum mechanics are described in more detail in Section 3, where the question of self-adjointness of the deformed Hamiltonian is answered. Section 4 is devoted to study the self-adjointness of the deformed spatial conjugate operator which is the generator of the QFT-Moyal-Weyl space-time. The last section investigates the intermediate relation between the well-definedness of the deformation and the self-adjointness of the resulting operator.

\section{Preliminaries}

\subsection{Warped Convolutions}
In  this section we write all   basic definitions and lemmas of the deformation known under the name of warped convolutions. For  proofs of the respective lemmas we refer the reader to the original works \cite{BLS, GL1}. \newline\newline 
The authors start their investigation with a $C^{*}$-dynamical system $(\mathcal{A},\mathbb{R}^n)$.
It consists of a $C^{*}$-algebra $\mathcal{A}$  equipped with a strongly continuous automorphic action of the group $\mathbb{R}^n$ which will be denoted by $\alpha$.
Furthermore, let  $\mathcal{B}(\mathscr{H})$ be the Hilbert space of bounded
operators on $\mathscr{H}$ and let  the adjoint action $\alpha$ be implemented  by the weakly continuous unitary representation $U$. Then,  it is argued that  since the unitary representation $U$ can be extended to the algebra $\mathcal{B}(\mathscr{H})$,   there is no lost of generality   when one proceeds to  the $C^{*}$-dynamical system  $(C^{*}\subset  \mathcal{B}(\mathscr{H}),\mathbb{R}^n)$.  Here $C^{*}\subset  \mathcal{B}(\mathscr{H})$  is the $C^{*}$-algebra of all operators on which $\alpha$ acts strongly continuously.
\newline\newline 
Hence, we start by assuming the
existence of a strongly continuous unitary group $U$ that is a representation of the additive
group $\mathbb{R}^{n}$, $n\geq2$, on some separable Hilbert space $\mathscr{H}$. 
Moreover, to define the deformation of operators belonging to a $C^{*}$-algebra $C^{*}\subset  \mathcal{B}(\mathscr{H})$, we consider
elements belonging to the sub-algebra $C^{\infty}\subset C^{*} $. The sub-algebra $C^{\infty}$ is
defined to be  the $*-$algebra of smooth elements (in the norm topology) with respect to $\alpha$, which
is the adjoint action of a weakly continuous
unitary representation $U$ of $\mathbb{R}^{n}$ given by  
\begin{equation*}
 \alpha_{x}(A)=U(x)\,A\,U(x)^{-1}, \quad x \in \mathbb{R}^{n}.
\end{equation*} 
Let $\mathcal{D}$ be the dense domain of vectors in $\mathscr{H}$ which transform
smoothly under the adjoint action of $U$.
Then, the warped convolutions  for operators  $A\in C^{\infty}$ are given by
the following definition.\\

\begin{definition}
 Let $\theta$ be a real skew-symmetric matrix on $\mathbb{R}^{n}$ w.r.t. a bilinear form, let $A\in C^{\infty}$  and let $E$ be
the spectral resolution of the unitary operator $U$.  Then, the corresponding \textbf{warped convolution} $A_{\theta}$
of $A$ is defined on the dense domain $\mathcal{D}$ according to
\begin{equation}\label{WC}
 A_{\theta }:=\int \alpha_{\theta x}(A)dE(x),
\end{equation}
where $\alpha$ denotes the adjoint action of $U$ given by $\alpha_{k}(A)=U(k)\,A\,U(k)^{-1}$.
\end{definition} 
The restriction in the choice of operators is owed to the fact that the deformation is performed
with operator valued integrals. Furthermore,  one can represent the warped
convolution of $
 {A} \in {C}^{\infty}$    by $\int \alpha_{\theta x}(A) dE(x)$, on the dense domain
$\mathcal{D}\subset\mathscr{H}$ of vectors smooth w.r.t. the action of $U$,  in terms of
strong limits 
\begin{equation*}
\int\alpha_{\theta x}(A) dE(x)\Phi=(2\pi)^{-n}
\lim_{\epsilon\rightarrow 0}
\iint  dx\, dy \,\chi(\epsilon x,\epsilon y )\,e^{-ixy}\, U(y)\, \alpha_{\theta x}(A)\Phi,  
\end{equation*}
where $\chi \in\mathscr{S}(\mathbb{R}^{n}\times\mathbb{R}^{n})$ with $\chi(0,0)=1$. In an intermediate step of showing this equivalence it was as well proven that $\mathcal{D}$ is \textbf{stable} under the deformation. 
The latter representation makes calculations and proofs concerning the existence of  integrals
easier. 
During this work we   use  both representations.  \\\\The following lemma shows first that the two different warped
convolutions are equivalent. Second, it demonstrates how the adjoint  of
the warped convoluted operator is obtained.\newline
\begin{lemma}\label{wcl2}
 Let $\theta$ be a real skew symmetric matrix on $\mathbb{R}^n$ and let $  {A} \in
\mathcal{C}^{\infty}$.
Then 
\begin{enumerate} \renewcommand{\labelenumi}{(\roman{enumi})}
 \item   $\int \alpha_{\theta x}(A)dE(x)=\int dE(x)\alpha_{\theta x}(A)$
 \item $\left(\int \alpha_{\theta x}(A)dE(x)\right)^{*}\subset\int \alpha_{\theta
x}(A^{*})dE(x)$
\end{enumerate}
\end{lemma} 
In the following lemma we introduce the deformed product, known as the Rieffel product \cite{RI}
by using the formalism of warped convolutions. The circumstance that the two are interrelated is due to
the fact that  warped convolutions supply isometric
representations of Rieffel's strict deformations of $C^{*}$-dynamical systems with
actions of $\mathbb{R}^{n}$.  The definition of the Rieffel product, given by warped convolutions, is used in the current work  to calculate
the deformed commutators. \newline
\begin{lemma}\label{dpl}
Let $\theta$ be a real skew-symmetric matrix on $\mathbb{R}^{n}$  w.r.t. a bilinear form, $  {A},   {B} \in
 C^{\infty}$  and let $\Phi\in\mathcal{D}$. Then, the product of two deformed operators $A,   B$ is given by
\begin{equation*}
 {A}_{\theta}   {B}_{\theta} \Phi= (A\times_{\theta}B)_{\theta}\Phi,
\end{equation*}
where the \textbf{deformed product} $\times_{\theta}$ is  the Rieffel product
on $ {C}^{\infty}$ defined as, 
\begin{equation}\label{dp0}
(A\times_{\theta}B )\Phi=(2\pi)^{-n}
\lim_{\epsilon\rightarrow 0}
\iint dx\, dy\,\chi(\epsilon x,\epsilon y )e^{-ixy} \alpha_{\theta x}(A)\alpha_{y}(B)\Phi.
\end{equation}
\end{lemma} 

\subsection{Kato-Rellich Theorem}In order to make this work self-contained we write in this Section the  Kato-Rellich theorem. In particular it gives a condition on the self-adjointness of an operator $A+B$ if $A$ is self-adjoint and $B$ symmetric.   
 \begin{definition}
An operator $B$ is called $A$-\textbf{{\textbf{bounded}}} if $\mathcal{D}(A)\subseteq\mathcal{D}(B)$ and if there are constants $a,b \geq 0$ such that
 \begin{align*} 
\|B\Psi\|\leq a\|A\Psi\|+b\| \Psi\|, \qquad \Psi\in\mathcal{D}(A)
\end{align*} The infimum of such an $a$ is called the 
relative bound of $B$ with respect to $A$. If the relative bound is zero, we say 
that $B$ is infinitesimally small with respect to $A$ and write $B < < A$.  
Sometimes it is more convenient to use, instead of the former, the following inequality
 \begin{align*} 
\|B\Psi\|^2\leq \tilde{a}^2\|A\Psi\|^2+\tilde{b}^2\| \Psi\|^2, \qquad \Psi\in\mathcal{D}(A).
\end{align*}
It is easy to show that the latter inequality follows from the former and vice versa, see \cite[Chapter X.6]{RS}.
\end{definition}
\begin{theorem}\label{krt}(\textbf{Kato-Rellich Theorem})
Suppose $A$ is (essentially) self-adjoint and $B$ is symmetric with $A$-bound less than one. Then $A+B$ is  (essentially) self-adjoint on $\mathcal{D}(A+B)=\mathcal{D}(A)$ and essentially self-adjoint on any core of A. \newline\newline If $A$ is bounded from below by $\gamma$, then $A+B$ is bounded from below by
 \begin{align*} 
\gamma-\text{max}(a|\gamma|+b,\frac{b}{a-1}).
\end{align*}
\end{theorem}
The former theorem has an extension to the case of relative bound one, \cite[Theorem X.14]{RS}.
\begin{theorem}\label{wt}(\textbf {W\"ust's Theorem})
Let $A$ be  self-adjoint and $B$ be symmetric with  $\mathcal{D}(A)\subseteq\mathcal{D}(B)$. Suppose that for some $b$ and all $\Psi\in\mathcal{D}(A)$, 
 \begin{align*}  
\|B\Psi\| \leq  \|A\Psi\| +{b} \| \Psi\| .
\end{align*}
 Then $A+B$ is \textbf{essentially} self-adjoint on $\mathcal{D}(A)$ or any core for $A$.

\end{theorem}

 \section{Deformations in QM}\label{c3}
In this Section we  study the self-adjointness of the minimal substituted Hamiltonian which is obtained by using warped convolutions in a QM context, \cite{Muc1}.  Let us first shortly summarize the content of this work.
We study deformations of the Laplace operator, i.e. the Hamiltonian of a non-interacting and non-relativistic  particle. The
generators of the deformation in the context of warped convolutions are chosen  
to be vector-valued functions of the coordinate operator. For the concrete proofs of the following lemmas and theorems we refer the reader to the original work. The deformation of the Hamiltonian is performed by working in the
standard realization of quantum mechanics, the so called \textbf{Schr\"odinger
representation}, \cite{BEH, RS1}. The operators 
$(P_{j},X_{k})$ satisfying the canonical commutation relations, are given in this representation  as essentially self-adjoint operators on the dense domain
$\mathscr{S}(\mathbb{R}^n)$. Here $P_{j}$ and $X_{k}$ are the closures of
$i {\partial}/{\partial x^j}$ and multiplication by $x_{k}$ on
$\mathscr{S}(\mathbb{R}^n)$, respectively.  
To apply the definition of  warped convolutions, we  demand to have self-adjoint operators that commute along their components. For this purpose we start with the following theorem, \cite[ Theorem VIII.6]{ RS1}.
\\
\begin{theorem}\label{tsa}
Let $\mathbf{Q}(.)$ be an unbounded real vector-valued Borel function on  $\mathbb{R}^n $ and let the dense domain $D_{\mathbf{Q}}$ be given as, 
\begin{equation*}
D_{\mathbf{Q}}=\{\phi| \int\limits_{ -\infty}^{\infty} |Q_j(\mathbf{x})|^2\,d(\phi,P_\mathbf{x}\phi) <\infty, \, \quad j=1,\dots, n\},
\end{equation*}
where $\{P_x\}$ are projection valued measures on $\mathscr{H}$. Then, $\mathbf{Q}(\mathbf{X})$ defined on $D_{\mathbf{Q}}$ is a self-adjoint operator. 
\end{theorem}
\begin{definition}  
Let $B$ be a real skew-symmetric matrix on $\mathbb{R}^{n}$ and 
let $\chi \in\mathscr{S}(\mathbb{R}^n\times\mathbb{R}^n)$ with $\chi(0,0)=1$. Moreover,  let 
$\mathbf{Q}(\mathbf{X})$ be given as in Theorem \ref{tsa}. Then, the warped convolutions of an   operator $A$ with  operator $\mathbf{Q}$, denoted as 
$A_{B,\mathbf{Q} }$  are
defined, in the same manner as in \cite{BLS}, namely
\begin{equation}\label{defcop}
A_{B,\mathbf{Q}  } := (2\pi)^{-n}
\lim_{\epsilon\rightarrow 0}
  \iint  \, d x \,  d y \, e^{-ixy}  \, \chi(\epsilon x,\epsilon y)V(y)
\alpha_{B x}(A) . 
\end{equation}
The automorphisms $\alpha$ are implemented by the adjoint action of the strongly
continuous unitary representation $V(x)=e^{ix_kQ^k}$ of $\mathbb{R}^n$ given by  
\begin{equation*}
 \alpha_{x}(A)=V(x)\,A\,V(x)^{-1}, \quad x \in \mathbb{R}^n.
\end{equation*} 
\end{definition} 
\subsection{Deforming the Hamiltonian}
For a non-interacting and non-relativistic quantum-mechanical particle the energy is described by the Hamiltonian given as
\begin{equation}\label{fh}
H_{0}= -\frac{P_jP^j}{2m}.
\end{equation} Due to the physical relevance, we restrict the
deformation to three space dimensions.  
The domain of self-adjointness
and the spectrum of  the free undeformed Hamiltonian $H_{0}$ (\cite{T}) is given in the following theorem.
\newline  
\begin{theorem}
 The free Schr\"odinger operator $H_{0}$ is self-adjoint on the domain
$\mathcal{D}(H_{0})$ given as 
\begin{equation*}
\mathcal{D}(H_{0})=H^2(\mathbb{R}^3)=\{\varphi \in
L^2(\mathbb{R}^3)||\textbf{P}|^2\varphi\in
L^2(\mathbb{R}^3)\},
\end{equation*}
and its spectrum is characterized by $\sigma(H_{0})=[0,\infty)$.
\end{theorem}$\,$\newline
  For the purpose of investigating the validity of the deformation Formula (\ref{defcop}) for unbounded operators, a dense domain $ \mathcal{E}\subseteq \mathscr{S}(\mathbb{R}^3)$  that fulfills additional requirements is introduced.  
\begin{lemma}\label{cnt0}
Consider the self-adjoint operator 

\begin{equation}\label{hs0}
 \mathbf{Q}(\mathbf{X}) =\mathbf{X}/ \vert\mathbf{X}\vert^{ n}, \qquad    n\in\mathbb{R}.
\end{equation}
 Then,  for all $n\in\mathbb{R}$  there exists a dense domain $ \mathcal{E}\subseteq \mathscr{S}(\mathbb{R}^3)$ such that
\begin{align} \label{hs}
\Vert 
\{P_{j},[\mathbf{Q},P^j]\}  \Phi \Vert<\infty, \qquad 
\Vert 
[\mathbf{Q},P_j][\mathbf{Q},P^j]  \Phi \Vert<\infty, \qquad \Phi\in  \mathcal{E} .
\end{align}
\end{lemma}
For $n\in \mathbb{R}^{+}$ we consider  the domain $\mathcal{E}$  which denotes the linear hull of the \textit{dense} vectors \cite[Theorem 3.2.5]{Th}
\begin{align*}
 \Phi(\mathbf{x})=x_1^{k_1}x_2^{k_2}x_3^{k_3}\,\exp{(-\frac{|\mathbf{x}|^2}{2} )},\qquad k_i=0,\,1,\,2,\,\dots.
\end{align*}  

\begin{proposition}\label{wcfh}
Let 
$\mathbf{Q}(\mathbf{X})$ be a self-adjoint  operator of the form
\begin{align*}
 \mathbf{Q}(\mathbf{X}) =\mathbf{X}/ \vert\mathbf{X}\vert^{ n}, \qquad    n\in\mathbb{R},
\end{align*}
and let $(H_{0})_{B,\mathbf{Q}}$ denote the deformed free Hamiltonian (see Formula (\ref{defcop})). Then, the deformation formula  for the unbounded operator
$H_{0}$, given as an oscillatory
integral, is   well-defined and the explicit result of the deformation is  
\begin{equation}\label{defh}
(H_{0})_{B,\mathbf{Q}}\Phi=
-\frac{1}{2m}\left
(P_j+i(B  Q )^k [Q_k,P_j]  \right)\left(P^j+i(B  Q )^r [Q_r,P^j]  \right)
\Phi,
\end{equation}
 where $\Phi \in \mathcal{E}\subseteq \mathscr{S}(\mathbb{R}^3)$.
\end{proposition}
The outcome of   deforming the Hamiltonian with vector-valued functions of the coordinate operator is  the so called \textbf{minimal
substitution}.  
 \newline\newline
In the  next proposition we give the resulting deformation of  the momentum operator.  
\begin{proposition} The deformation of the unbounded momentum operator, given as an oscillatory
integral, is well-defined.  Moreover,  the explicit result of the deformation is given as
\begin{equation}\label{defp}
 {P}^j_{B,\mathbf{Q}}\Phi=\left(
 P^j+i(B  Q )^k [Q_k,P^j]  \right)
\Phi,  \qquad  \forall   \Phi \in \mathcal{E}\subseteq \mathscr{S}(\mathbb{R}^3).
\end{equation}
 
\end{proposition}
In order to settle the question of arbitrariness in the definition of the deformed free Hamiltonian the next theorem is essential. The arbitrariness stems from the view point that the deformed Hamiltonian could  be defined as the scalar product of the deformed momentum operators.

\begin{theorem}\label{d1}
The scalar product of the deformed momentum vectors is equal to the deformed free Hamiltonian (see Equation \ref{defh}),
i.e. 
\begin{equation*}
 (H_{0})_{B,\mathbf{Q}}\Psi=-\frac{1}{2m}P_j^{B,\mathbf{Q}}P_{B,\mathbf{Q}}^j
\Psi, \qquad \Psi \in\mathcal{E}\subseteq \mathscr{S}(\mathbb{R}^3).
\end{equation*}
 
\end{theorem} 
\subsection{Self-adjointness of the Deformed Hamiltonian}
In this Section we shall give two different proofs for the (essential) self-adjointness of the deformed operators. The first proof is done by using Lemma \ref{cnt0}   and the Kato-Rellich theorem. In particular, our deformed Hamiltonian $(H_0)_{B,Q}$  is the sum of the free Hamiltonian and an additional term, i.e. 
\begin{align*} 
(H_{0})_{B,\mathbf{Q}}\Phi&=
-\left
(P_j+i(B  Q )^k [Q_k,P_j]  \right)\left(P^j+i(B  Q )^r [Q_r,P^j]  \right)
\Phi\\&
= H_0\Phi-i\{(B  Q )^k [Q_k,P_j],P^j\}\Phi+(B  Q )^k [Q_k,P_j](B  Q )^r [Q_r,P^j]\Phi\\&=:
H_{0}\Phi+V(\mathbf{X}, \mathbf{P})\Phi.
\end{align*}
\begin{remark}During the whole Section we set the mass equal to one half. This is just a matter  of convenience and convention. Moreover, without loss of generality, one can choose the  skew-symmetric matrix $B$  to have the  
form  $B_{ij}= \varepsilon_{ijk}B^j$, where $\varepsilon_{ijk}$ is the three dimensional
epsilon-tensor. Then, we are able to derive the following inequality,
\begin{equation}\label{iny}
  \vert {(By)_i {e}^i}\vert\leq \sqrt{2} \vert \mathbf{B}\vert  \vert \mathbf{y}\vert.
\end{equation}
This is easily seen by using Cauchy-Schwarz and the inequality $\vert a\vert -\vert
b\vert\leq\vert a\vert+\vert b\vert$.
 \end{remark}$\,$\newline
To simplify the forthcoming calculations let us give  general formulas for the commutators 
\begin{align}\label{hs1}
[P_j,  \vert\mathbf{X}\vert^{-n}]= i \,n\, 
 X_j\vert\mathbf{X}\vert^{ -(n+2)} ,
\end{align}
\begin{align}\label{hs2}
[P_j,   X_k /\vert\mathbf{X}\vert^{n}]= i\left(\eta_{jk}+
 n\,X_kX_j/\vert\mathbf{X}\vert^{ 2}\right)\vert\mathbf{X}\vert^{-n}.
\end{align}   
Hence in order to show the self-adjointness of the deformed Hamiltonian it suffices to show that the $H_0$- bound  of the operator valued function $V(\mathbf{X}, \mathbf{P})$ is less than one. 
To simplify the calculations  we look at the explicit term $V(\mathbf{X}, \mathbf{P})$ by plugging $\mathbf{X}/ \vert\mathbf{X}\vert^{ n}$ for $Q$. For the simplifications we shall use Equations $(\ref{hs1})$ and $(\ref{hs2})$. Let us start with the first term,
\begin{align*} 
\{(B  Q )^k [Q_k,P_j],P^j\}&=-i\{(B X    )^k   \vert\mathbf{X}\vert^{-n}
\left(\eta_{kj}+
 n\,X_kX_j/\vert\mathbf{X}\vert^{ 2}\right)\vert\mathbf{X}\vert^{-n}
,P^j\}\\&=-i\{(B    {X} )_j   \vert\mathbf{X}\vert^{-2n}
,P^j\}\\&=-2i  \vert\mathbf{X}\vert^{-2n}
(B    {X} )_j P^j,
\end{align*}
where in the second line  we used the skew-symmetry of the matrix $B$, i.e. $(B    {X} )^k {X}_k=B^{kj}X_jX_k=0$. Moreover,  in the last line the canonical commutation relations and the skew-symmetry of the matrix $B$ were used in order to prove that the commutator   $[(B    {X} )_j   \vert\mathbf{X}\vert^{-2n},P^j]$ is equal to zero. The last term of $V(\mathbf{X}, \mathbf{P})$ is fairly easy and can be simplified as,
\begin{align*} 
(B  Q )^k [Q_k,P_j](B  Q )_r [Q^r,P^j]&= -(B X    )^k (B X    )^r \eta_{jk} \eta^{rj}\vert\mathbf{X}\vert^{-4n}\\
&= -(B X    )_j (B X    )^j  \vert\mathbf{X}\vert^{-4n},
\end{align*}
where here as well we used  the skew-symmetry of the matrix $B$. Therefore, we can write the operator $V(\mathbf{X}, \mathbf{P})$ in its simplified form as
\begin{align}  
 V(\mathbf{X}, \mathbf{P})= -2 \vert\mathbf{X}\vert^{-2n}
(B    {X} )_j P^j-(B X    )_j (B X    )^j    \vert\mathbf{X}\vert^{-4n}
\end{align}
\begin{lemma}
The operator valued function $V(\mathbf{X}, \mathbf{P})$ is a symmetric operator on the dense domain $\mathcal{E}\subseteq \mathscr{S}(\mathbb{R}^3)$.
\end{lemma}
\begin{proof}
For the first term of the object symmetry is proven straight forward, since the partial differentiation does not affect the $X$-dependent terms. This fact follows from the skew-symmetry of $B$. Furthermore, since the second term of the operator $V(\mathbf{X}, \mathbf{P})$ is merely a real-valued Borel function of the coordinate operator symmetry (and even self-adjointness) follows from Theorem \ref{tsa}.
\end{proof}

\begin{theorem}
The symmetric operator   $V(\mathbf{X}, \mathbf{P})$ has an $H_0$-bound less than one for vectors of the dense domain $\mathcal{E}\subseteq \mathscr{S}(\mathbb{R}^3)$. Therefore, the deformed Hamiltonian $(H_{0})_{B,\mathbf{Q}}$ is a self-adjoint operator. 
\end{theorem}
\begin{proof}
The symmetry of the operator $V(\mathbf{X}, \mathbf{P})$  follows from the former lemma. Next, we start the proof by looking at the second much simpler term of the $V(\mathbf{X}, \mathbf{P})$, i.e. we want to show that 
\begin{align*} 
\| (B X    )^k (B X    )_k  \vert\mathbf{X}\vert^{-4n}\Phi\|\leq a\|H_0\Phi\|+b\|\Phi\|,
\end{align*}
for all $\Phi \in \mathcal{E}\subseteq \mathscr{S}(\mathbb{R}^3)$ and $a,b>0$.
\begin{align*} 
\| (B X    )^k (B X    )_k  \vert\mathbf{X}\vert^{-4n}\Phi\|^2&
= \int d^3x | (B x    )^k (B x    )_k  \vert\mathbf{x}\vert^{-4n}\Phi(x) |^2\\&\leq   
\int d^3x \underbrace{ 4\vert \mathbf{B}\vert^4  \vert\mathbf{x}\vert^{-8n+4}|\Phi(x) |}_{=:\chi(x)}  |\Phi(x) |  \\&\leq 
\|\chi \|_{1}
\|\Phi\|_{\infty} =
 \frac{\|\chi \|_{1} }{\|\Phi\|_{\infty}} \|\Phi\|^2_{\infty}
\\&\leq c_1\|\Phi\|^2_{\infty}\leq a^2  \|H_0\Phi\|^2+b^2\|\Phi\|^2
,
\end{align*}
where the constant $c_1$ is finite since for negative $n$ and $\Phi\in\mathscr{S}(\mathbb{R}^3)$ the resulting integral is finite and for positive $n$ the integral exists as well which was proven rigorously in \cite[Chapter 1, Section 3.6]{GS1}. Let us further remark on the finiteness of the $L^{ \infty}$ norm of $\Phi$. Since
 $\Phi\in\mathscr{S}(\mathbb{R}^3)$ it is a subset of $L^p(\mathbb{R}^3)$ for $p$ such that $1\leq p \leq \infty$. The $L^{ \infty}$ norm of $\Phi$ can only be zero if $\Phi$ is zero, but in that case the inequality that we are proving would hold without any further work. Furthermore in the last lines we used the relation (\ref{iny}), the H\"older-inequality and  \cite[Theorem X.15, Equation X.22]{RS}. Now since for any $a<1$ there is a $b$ such that the former inequality holds the first part of $V(\mathbf{X}, \mathbf{P})$ has an $H_0$-bound less than one. Next, we prove the boundedness for the second part of $V(\mathbf{X}, \mathbf{P})$. This  follows similar arguments as the proof of \cite[Theorem X.22]{RS},
\begin{align*} 
\|2   \vert\mathbf{X}\vert^{-2n}
(B    {X} )_j P^j\Phi\| &
= \left(\int d^3x | 2\,\vert\mathbf{x}\vert^{-2n} (B x    )_j \partial^j \Phi(x) |^2 \right)^{1/2}
\\&\leq  \left(
\int d^3x\underbrace{4\vert \sqrt{2}\vert \mathbf{B} | \vert\mathbf{x}\vert^{-4n+1} (B x    )_j \partial^j \Phi(x) |}_{=:\Psi(x)} |  \vec{\partial}  \Phi(x) | \right)^{1/2}
\\&\leq  \|\Psi\|^{\frac{1}{2}}_{\frac{4}{3}}
\|\vec{\partial}  \Phi\|^{\frac{1}{2}}_{4}=
  \frac{\|\Psi\|^{\frac{1}{2}}_{\frac{4}{3}}} {\|\vec{\partial}  \Phi\|^{\frac{1}{2}}_{4} } \|\vec{\partial}  \Phi\| _{4}\leq
c_2
\|\vec{P}  \mathscr{F}(\Phi)\|_{\frac{4}{3}}\\& =c_2
\|(1+|\vec{P}|)^{-\alpha}(1+|\vec{P}|)^{\alpha}\vec{P}  \mathscr{F}(\Phi)\|_{\frac{4}{3}}\\&\leq c_2
\|(1+|\vec{P}|)^{-\alpha}\| _{ {4} }\| (1+|\vec{P}|)^{\alpha}\vec{P}  \mathscr{F}(\Phi)  \|   , 
\end{align*}
where $\alpha$ was chosen to be a fixed number in $(3/4, 1)$, which is a condition necessary to secure finiteness of $\|(1+|\vec{P}|)^{-\alpha}\|_{{4}}$. 
In the last lines we used  the H\"{o}lder and Hausdorff-Young inequalities and similar arguments hold for the finiteness of $c_2$ as in the former part. 
In particular for  $n\leq 0$ the resulting vector $\Psi$ is in the Schwartz space and therefore the integral is finite. Now for positive $n$ we need the argument that multiplication and differentiation leave the domain $\mathcal{E}$ invariant  (see \cite[Theorem 3.2.5]{Th}), therefore we are left with the argument that the integral exists for all $n$ and vectors in  $\mathcal{E}$. Moreover, the $L^{ 4}$ norm of $\vec{\partial}  \Phi$ can only be zero if $\Phi$ is zero (since $\Phi\in \mathcal{E}$), as before in that case the inequality that we are proving holds as well. Hence, for the relevant term remaining we use the fact that for any $a>0$ there is a $b$ such that  
\begin{align*} 
 \| (1+|\vec{P}|)^{\alpha}\vec{P}  \mathscr{F}(\Phi)  \|  &\leq  \| (b+a|\vec{P}|^2) \mathscr{F}(\Phi)  \| 
\\&\leq a\| \Delta \Phi  \| +  b  \| \Phi \|.
\end{align*}
\end{proof}
\begin{remark}
It is important to point out that for the two parts of the operator $V(\mathbf{X}, \mathbf{P})$ we have proven the relative $A$-bound. Now since the constant $a$ can be chosen arbitrarily for both parts (see former proof) the resulting overall constant can be chosen to be smaller than one. 
\end{remark}

Another proof is given by arguing that the deformed momentum operator is a self-adjoint operator. Now since we have proven that the square of the deformed momentum is the deformed Hamiltonian, self-adjointness follows. 
\\
\begin{theorem}
The deformed momentum operator defined on the dense domain  $\mathcal{E}\subseteq \mathscr{S}(\mathbb{R}^3)$,  given as
 \begin{equation*} 
 {P}^j_{B,\mathbf{Q}} =\left(
 P^j+i(B  Q )^k [Q_k,P^j]  \right),
\end{equation*}
is an essentially self-adjoint operator. Therefore, the deformed Hamiltonian 
 \begin{equation*} 
H_{B,\mathbf{Q}}= \vec{P}_{B,\mathbf{Q}} ^2
\end{equation*}
is an essentially self-adjoint operator on the dense domain  $\mathcal{E}\subseteq \mathscr{S}(\mathbb{R}^3)$.
\end{theorem}
\begin{proof}
From \cite[Theorem 3.2.5]{Th} we know that the momentum operator is an essentially self-adjoint operator
on the dense domain  $\mathcal{E}\subseteq \mathscr{S}(\mathbb{R}^3)$. Hence, it remains to solve the question of self-adjointness for the second part of the deformed momentum operator. The second part is nothing else than a real-valued function of the coordinate operator, which is also self-adjoint on  $\mathcal{E}$, and therefore (see Theorem \ref{tsa}) by itself it is an essentially self-adjoint operator. 
From Theorem \ref{d1} we know that the deformed Hamiltonian consists of the square of the deformed momentum. Hence as before by the virtue of Theorem \ref{tsa} essential self-adjointness of the deformed Hamiltonian follows. 
\end{proof}

\section{Deformation in QFT}

\subsection{Bosonic Fock space} The ($d=n+1 $)-dimensional relativistic  bosonic Fock 
space is defined in the following. A particle with momentum $\mathbf{p} \in \mathbb{R}^n$ has in   the massless case the energy $p_0=\omega_{\mathbf{p}}=+\sqrt{\mathbf{p}^2}$. Moreover the Lorentz-invariant measure is given by   $d^n\mu(\mathbf{p} )=d^n\mathbf{p}/( {2\omega_{\mathbf{p}}})$.
\begin{definition}\label{bf}
 The \textbf{bosonic Fock space} $\mathscr{F^{+}({H})}$ is defined 
as in \cite{BR}:
\begin{equation*}
\mathscr{F^{+}({H})}=\bigoplus_{k=0}^{\infty}\mathscr{H}_{k}^{+},
\end{equation*}
where $\mathscr{H}_{0}=\mathbb{C}$ and  the symmetric $k$-particle subspaces are given as
  \begin{align*}
\mathscr{H}_{k}^{+} =\{\Psi_{k}: \underbrace{\partial V^{+}  \times  \dots \times 
\partial V^{+}}_{k-times} \rightarrow \mathbb{C}\quad \mathrm{symmetric}
| \left\Vert  \Psi_k \right\Vert^2  <\infty\},
\end{align*}
with $
  \partial V^{+}:=\{p\in \mathbb{R}^{d}|p^2=0,p_0>0\}.$
  \end{definition}
  The particle annihilation and creation operators $a,a^{*}$ of the  bosonic
Fock space satisfy the following commutator relations
  \begin{align}\label{pccr}
[a_c(\mathbf{p}), a_c^{*}(\mathbf{q})]=2\omega_{\mathbf{p} 
}\delta^n(\mathbf{p}-\mathbf{q}), \qquad
[a_c(\mathbf{p}), a_c(\mathbf{q})]=0=[a_c^{*}(\mathbf{p}), a_c^{*}(\mathbf{q})].
\end{align} 
By using  $a_c,a_c^{*}$   the particle number operator and the momentum operator are  defined in the following manner,
 \begin{equation}\label{pcaopm}
N=\int d^n\mu(\mathbf{p}) a_c^{*}(\mathbf{p})a_c(\mathbf{p}), \qquad P_{\mu}=\int
d^n\mu(\mathbf{p})\,  p_{\mu} a_c^{*}(\mathbf{p})a_c(\mathbf{p}).
\end{equation}

\begin{remark}
Since it is simpler in what follows we shall use from now on the terms of the annihilation and creation operators   in the noncovariant  representation, i.e. 
\begin{equation*} 
 {a}(\textbf{p}):= \frac{{a}_c (\textbf{p})}{\sqrt{2\omega_{\mathbf{p}}}},\qquad  {a}^* (\textbf{p}):= \frac{{a}_c^* (\textbf{p})}{\sqrt{2\omega_{\mathbf{p}}}}.
\end{equation*} 
\end{remark}

  \subsection{Deforming the Coordinate Operator}
In the context of QM the deformation of the coordinate operator  with the
momentum operator gave us the quantum plane of  the Landau-quantization, (see \cite[Lemma 4.3]{Muc1}). 
\\\\Now by following the idea found in a QM context we calculate the  commutator of the deformed {quantum field theoretical} spatial
conjugate operator  by using the momentum   for deformation. The resulting  quantum spacetime is called the \textbf{QFT-Moyal-Weyl spacetime}, see \cite{Muc3}. Now in order to give the spatial coordinate operator in a {QFT-context} we took the quantum mechanical unitary equivalence of the momentum and coordinate operator, given by the Fourier transformation, into account and performed a second quantization.
\begin{lemma}\label{xsa} 
The spatial coordinate operator $\mathbf{X}$ on the one-particle subspace is an essentially self-adjoint operator on the dense domain $ \mathscr{S}
(\mathbb{R}^n)$. Furthermore, its second-quantization  is an essential self-adjoint operator on the dense domain $\mathscr{S}
(\mathbb{R}^n)_{\otimes}\subset \mathscr{F ({H})}$, which is the set of $\psi=\{\psi_{0},\psi_{1},\cdots\}$ such
that
$\psi_{k}=0$ for $k$ large enough and $\psi_{k}\in \bigotimes_{i=1}^k \mathscr{S}
(\mathbb{R}^n)$ for each $k$. Moreover, it has the following bosonic Fock space representation  
\begin{equation}\label{uexp3}
X_j=-i\int d^n \mathbf{p}\,  {a}^*(\textbf{p}) \frac{\partial}{\partial p^j}  {a}(\textbf{p}).
\end{equation}

\end{lemma}

\begin{lemma}\label{qftdefx}
The deformed operator that is obtained by deformation with warped convolutions, by using the momentum operator, is represented on the dense domain $\mathscr{S}
(\mathbb{R}^n)_{\otimes}$   by
\begin{equation*} 
X^{j}_{\theta}=X^{j} +\left(\theta P \right)^{0}V^j - \left(\theta P \right)^{j}N.
\end{equation*}
\end{lemma}

In the next step we give the commutator of the deformed coordinate operator. In order to make  relativistic corrections more apparent   we do not set the speed of light $c$ equal to one.  \newline
\begin{theorem}\label{qmwxc} \textbf{QFT-Moyal-Weyl.} The deformed commutator of the coordinate operators represented on the dense domain $\mathscr{S}
(\mathbb{R}^n)_{\otimes}$ is given by 
\begin{align*} 
  i\widehat{\theta}_{ij}=  -2 i\left(
\theta_{0i} V_j /c-\theta_{0j} V_i/c \right)N
-2i\theta_{ij}N^2,
\end{align*}
where $\mathbf{V}$ is the velocity operator given by
\begin{equation}\label{vel}
V_j= \int d^n\mathbf{p}\,\frac{p_{j}}{\omega_{\textbf{p}}}\,{a}^*(\mathbf{p})
{a}(\mathbf{p}).
\end{equation}
\end{theorem}
\subsection{Self-adjointness of the Deformed Coordinate Operator}
 An unanswered, but important, question in the context of the QFT-Moyal-Weyl is the self-adjointness of the deformed spatial coordinate operator. In this section we shall tackle this problem.
\begin{lemma}
The part  of the spatial operator obtained by deformation, i.e. $\left(\theta P \right)^{0}V^j - \left(\theta P \right)^{j}N$, is a symmetric operator for vectors of  the dense domain $ \mathscr{S}
(\mathbb{R}^n)_{\otimes}$.
\end{lemma}

\begin{proof}The proof of symmetry, as always, is fairly easy and it is done by considering the operator in the scalar product w.r.t. the bosonic Fock space. Since we only have real functions of the momentum operator and the vectors are from  the Schwartz space, symmetry follows. 
\end{proof}
After showing symmetry the task remains to show the $\mathbf{X}$-bound  of the outcome from deformation. This is the result of the following theorem.
\begin{theorem} 
The deformed  coordinate operator   $\mathbf{X} _{\theta}$, obtained by warped convolutions, is a self-adjoint operator on the dense domain $\mathscr{S}
(\mathbb{R}^n)_{\otimes}$.
\end{theorem}

\begin{proof}
The proof is done by using the Kato-Rellich theorem. Hence, we show  that $\mathbf{X} _{\theta}$ has a  $\mathbf{X}$-bound less than one for vectors of the dense domain $\mathscr{S}
(\mathbb{R}^n)_{\otimes}$. Hence we want to show that
\begin{align*} 
\| (\mathbf{X} _{\theta}- \mathbf{X})\Phi\| \leq  {a} \|\mathbf{X}\Phi\| + {b} \|\Phi\| .
\end{align*}
In order to proceed we look at the explicit term
\begin{equation*} \begin{split}
\| \left(\left(\theta P \right)^{0} {V}^j  - {\left(\theta P \right)}^j N \right)e_j \Phi\| &\leq  
\|  \left(\theta P \right)^{0}V^j e_j \Phi\|  +\| \left(\theta P \right)^{j}e_j N \Phi\| \\&\leq
|\vec{\theta}_0 |\| \mathbf{P}\mathbf{V} \Phi\| +\| \left(\vec{\theta}_0P_0+ \theta^{jk}e_jP_k \right) N\Phi\| \\&\leq
|\vec{\theta}_0 |\| \mathbf{P}\mathbf{V}\Phi\| +\|  \vec{\theta}_0 P_0N\Phi\| +\sqrt{2} | \vec{\theta}| \| \mathbf{P}N\Phi\| 
\\&\leq { \left(|\vec{\theta}_0|\left(   \| \mathbf{P}^2\mathbf{V}^2\Phi\|^{\frac{1}{2}} + \| P^2_0N^2
\Phi\|^{\frac{1}{2}} \right)
+\sqrt{2}  |\vec{\theta}| \| \mathbf{P}^2N^2
\Phi\|^{\frac{1}{2}}\right)} \\
 &\quad\times
\|\Phi\|^{\frac{1}{2}}\\&=: \frac{\|\Psi\|^{\frac{1}{2}}}{\|\Phi\|^{\frac{1}{2}}}\|\Phi\| \leq b \|\Phi\|,
\end{split} 
\end{equation*}
where we used the Cauchy-Schwartz Inequality, for the spatial part of $\theta$  Equation (\ref{iny}),  for  the temporal part $\theta^{0j}e_j=:\vec{\theta}_0$
and the fact that our vectors belong to the Schwartz-space prove that the constant $b$ is finite. As in the QM part we argue that for $\Phi =0$  the inequality holds without further evaluation and hence this case is ignored.  Note that we used the fact that the momentum operator, the velocity operator and the particle number operator commute and that the resulting operators evaluated on 
$\Phi \in \mathscr{S}
(\mathbb{R}^n)_{\otimes}$ are finite. 
\end{proof}

\section{Deformation of a General Unbounded Operator}
In this Section we shall prove the self-adjointness of a deformed self-adjoint unbounded operator. The proof is done by assuming the well-definedness of  the deformation which is given as an oscillatory integral. Now the deformation formula, given by warped convolutions, is only well-defined in the strong operator
topology for a subset of \textbf{bounded operators} that are \textbf{smooth} w.r.t. unitary
representation
$U$ of
$\mathbb{R}^n$. In view of the fact that we deal with unbounded operators  we 
have to show that the deformation formula, given as an oscillatory
integral, is well-defined.
For the subsequent discussion, let us introduce the notion of an oscillatory integral,
(\cite{H}, Section 7.8, Equation 7.8.1).
\begin{definition}\label{oipf}
 Let $X\subset\mathbb{R}^n$ be open and let $\Gamma$ be an open cone on $X\times(\mathbb{R}^N
\textbackslash \{0\})$ for some $N$. This means that $\Gamma$ is invariant under multiplication by
positive scalars of components in $\mathbb{R}^N$. We shall say that a function $\phi\in
C^{\infty}(\Gamma)$ is a \textbf{phase function} in $\Gamma$ if
 \begin{itemize}
  \item $\phi(x,ty)=t\phi(x, y)$ if $(x,y)\in \Gamma, t>0$.
\item Im$\phi\geq0$ in $\Gamma$,
\item d$\phi\neq0$ in $\Gamma$.
 \end{itemize}
Then an integral of the form 
\begin{equation*}
 \int \, e^{i\phi(x,y)}b(x,y) \, dy,
\end{equation*}
is called an \textbf{oscillatory integral}.
\end{definition}
Another important notion in our subsequent discussion is that of a symbol   (\cite{H}, Section
7.8, Definition 7.8.1). 
\begin{definition}
 Let $m$, $\rho$, $\delta$, be real numbers with $0<\rho\leq1$ and $0\leq \delta <1$. Then we
denote by $S^{m}_{\rho,\delta}(X\times \mathbb{R}^n)$, the set of all $b\in C^{\infty}(X\times
\mathbb{R}^n)$ such that for every compact set $K\subset X$ and all $\gamma, \beta$ the estimate 
\begin{equation*}
 |\partial_{x}^{\gamma}\partial_{y}^{\beta}b(x,y)|\leq
C_{\gamma,\beta,K}(1+|x|)^{m-\rho|\gamma|+\delta|\beta|},\qquad x\in K, \,\, y\in
\mathbb{R}^n,  
\end{equation*}
is valid for some constant $C_{\gamma,\beta,K}$. The elements  $S^{m}_{\rho,\delta}$ are called
symbols of order $m$ and type $\rho,\delta$.
\end{definition}
By using the former definitions it can be shown \cite[Section
7.8, Theorem 7.8.2]{H}, (\cite{LW}), (\cite{Jo}), that if $m<-n+1$ the
oscillatory integral converges  to a well-defined function.  In the case  
$m\geq -n+1$, the oscillatory integral can still be defined in a distributional manner.  The
cases considered in this work belong to the second class.  
To prove that the deformation formula (\ref{defcop}) holds in the case of an unbounded operator
$A$, defined on a dense domain $\mathcal{D}(A)\subset \mathscr{H}$ of some separable Hilbert space $\mathscr{H}$, let us
consider the deformed operator $A_{\theta}$ as follows
\begin{align*}
\langle \Psi,A_{\theta}\Phi\rangle&=
(2\pi)^{-n}
\lim_{\epsilon\rightarrow 0}
  \iint  \, dx \,  dy \, e^{-ixy}  \, \chi(\epsilon x,\epsilon
y) {\langle \Psi, 
U(y)\alpha_{\theta x}(A)\Phi\rangle} \nonumber \\&
=:
(2\pi)^{-n}
\lim_{\epsilon\rightarrow 0}
  \iint  \, dx\,  dy \, e^{-ixy}  \, \chi(\epsilon x,\epsilon
y) \, b(x,y)
\end{align*}
for $\Psi, \Phi \in \mathcal{D}^{\infty}(A):=\{
\Phi\in\mathcal{D}(A)| U(x)\Phi \in \mathcal{D}(A)$  is smooth in $\|\cdot\|_{\mathscr{H}}
\}$. 
 Thus to prove that the expression is well-defined, $b(x,y)$ has to  be a symbol.  
\begin{lemma}\label{lfhs}
Assume that the derivatives of the adjoint action of $A$ w.r.t. the unitary operator $U$ are polynomially bounded on vectors in $\mathcal{D}^{\infty}(A)$, i.e.
\begin{equation}\label{pb}
 \|\partial_{x}^{\gamma}\alpha_{\theta x}(A)\Phi\| \leq  C_{\gamma}(1+|x|)^{m-\rho|\gamma|},\qquad \forall \Phi \in
\mathcal{D}^{\infty}(A).
\end{equation}
  Then, $b(x,y)$ belongs to the symbol class $S^{m}_{\rho,0}$ for $\Psi, \Phi \in
\mathcal{D}^{\infty}(A)$ and therefore the  deformation, via warped
convolution, of the unbounded operator $A$ is given as a well-defined
oscillatory integral.
\end{lemma} 
\begin{proof}
For the derivatives of the  scalar product $b(x,y)$ the following estimates hold

\begin{align*} \partial_{x}^{\gamma}\partial_{y}^{\beta}
{\langle \Psi, 
U(y)\alpha_{\theta x}(A)\Phi\rangle}&\leq \| (-iX)^{\beta} \Psi\| \| 
 \partial_{x}^{\gamma}\alpha_{\theta x}(A)\Phi\|\\&\leq\underbrace{
 \| (-iX)^{\beta} \Psi\|  C_{\gamma}}_{ =:C_{\gamma,\beta,K}}(1+|x|)^{m-\rho|\gamma|},
\end{align*}
where in the last lines we used Assumption (\ref{pb}). To prove the second part of the statement we first define $m(\rho):=m-\rho|\gamma|$ and use the former inequality, i.e. 
\begin{align*}
&(2\pi)^{-n}
\lim_{\epsilon\rightarrow 0}
  \iint  \, dx\,  dy \, e^{-ixy}  \, \chi(\epsilon x,\epsilon
y) \,\partial_{x}^{\gamma}\partial_{y}^{\beta} b(x,y)
\\ &\leq(2\pi)^{-n}
\lim_{\epsilon\rightarrow 0}
  \iint  \, dx\,  dy \, e^{-ixy}  \, \chi(\epsilon x,\epsilon
y) \,C_{\gamma,\beta,K}(1+|x|)^{m(\rho)}
 \\&\leq  (2\pi)^{-n} C_{\gamma,\beta,K}
\lim_{\varepsilon_1\rightarrow 0}  \left(
\int dx \lim_{\varepsilon_2\rightarrow 0} 
\left(\int dy  e^{-ixy}
\chi_2(\varepsilon_2 y)\right)\,\chi_1(\varepsilon_1  x)\,
 \right)(1+|x|)^{m(\rho )}
\\
&=       (2\pi)^{-n/2}C_{\gamma,\beta,K}
\lim_{\varepsilon_1\rightarrow 0}  \left(
\int dx \,
\delta( {x} )\,\chi_1(\varepsilon_1  x)   (1+|x|)^{m(\rho )} \right) =C_{\gamma,\beta,K}.
\end{align*}
For another proof of the second statement we refer the reader to \cite{LW} and \cite[Theorem 1]{AA}.
\end{proof}
\begin{lemma}\label{lfhs1}
The operator $A_{\theta}-A_0$ is a symmetric operator on the dense domain $\mathcal{D}^{\infty}(A)$.
\end{lemma} 
\begin{proof}
Since the undeformed operator $A_0$ is a self-adjoint operator for all vectors of $\mathcal{D} (A)$ and $\mathcal{D}^{\infty}(A)$ is a dense subset of the domain, symmetry follows. For the first part of the operator, i.e. $A_{\theta}$ symmetry follows easily from Lemma \ref{wcl2}. 
\end{proof}$\,$\newline
Next, we come to the most important result of this current work.
\begin{theorem}
Let the operator $A_{\theta}$   fulfill the assumptions made in Lemma \ref{lfhs}. Then the  symmetric operator $A_{\theta}-A_0$  has a relative $A_0$-bound one and therefore   $A_{\theta}$ is an essentially self-adjoint operator  on $\mathcal{D}^{\infty}(A)$. 
\end{theorem} 
\begin{proof}
The symmetry of  the operator under consideration, i.e. $A_{\theta}-A_0$, follows from Lemma \ref{lfhs1}. Next we turn our attention to the  relative $A_0$-bound. To do so, we first give the an inequality for the $L^2$-norm of the deformed operator. For the following let us look closer at
 \begin{align*} 
{\langle U(y')\alpha_{\theta x'}(A)\Phi  , 
U(y)\alpha_{\theta x}(A)\Phi\rangle} 
&\leq\| \alpha_{\theta x'}(A)\Phi\|\| \alpha_{\theta x }(A)\Phi\|\\&\leq
C_0^2  (1+|x'|)^{m}(1+|x |)^{m},
\end{align*}
where in the last line we used the assumption of polynomial boundedness (see Equation \ref{pb}). Hence by using this result, similar to the calculation  in the proof of Lemma \ref{lfhs} we obtain
 \begin{align*} 
\|  A_{\theta}  \Phi\|^2 &\leq C_0^2,
\end{align*}
where in the last lines we carried out a calculation and used the notation as the one in the proof of Lemma \ref{lfhs}.
Let us now use the former inequality to prove the  relative $A_0$-bound of the deformed operator,
 \begin{align*} 
\| (A_{\theta}-A_0)\Phi\| &\leq    \|A_0\Phi\| +   \|A_{\theta}\Phi\| 
\\&\leq \|A_0\Phi\| +\underbrace{\frac{C_0}{\|\Phi\|}}_{=:b}\|\Phi\| .
\end{align*}
By using W\"ust's Theorem (see Theorem \ref{wt}) the essential self-adjointness  of the operator $A_{\theta}$ follows. 
\end{proof}
Hence, it is interesting to note that for the deformation to be well defined we had to impose the polynomial boundedness assumption. However by doing so we were able to prove the essential self-adjointness of the deformed operator.
 
\section{Conclusion and Outlook}

In this work we gave two proofs for the self-adjointness of the deformed Hamiltonians given in \cite{Muc1}.  One proof was done by using the famous Kato-Rellich theorem and the second one was performed by the use of functional analytic arguments. Moreover, we have proven the self-adjointness of the second-quantized deformed spatial coordinate operator \cite{Muc3} that generates the QFT-Moyal-Weyl. 
\\\\
For deformed fields,  self-adjointess is usually shown by Nelson's analytic vector theorem. Hence in that particular case essential self-adjointness is proven on a dense set of analytic vectors. This proof can be generalized to a more complex structure if one uses operators that are unitarily equivalent to the momentum operator for deformation. Among other properties this is the subject of investigation in \cite{Muc4}.
\\\\
However, the proof for  general unbounded deformed operators  was still open. To solve the general problem for an arbitrary self-adjoint unbounded operator we used the decay properties of the oscillatory integral. In particular this means that we have to have a sufficiently decaying oscillatory integral in order for the deformation to be well-defined. By using the inequalities that follow from this particular proof,  W\"ust's theorem can be used to prove self-adjointness.  
\\\\
Note that  the operators in the QM and  QFT case belong to the  class   satisfying the assumption of polynomial boundedness w.r.t.  derivatives of the adjoint action.

 \bibliographystyle{alpha}
\bibliography{allliterature}

\begin{thebibliography}{{Muc}15b}

\bibitem[Ala]{A}
Sabina Alazzawi.
\newblock {Deformations of Fermionic Quantum Field Theories and Integrable
  Models}.
\newblock {\em Lett. Math. Phys. 103 (2013) 37-58}.

\bibitem[And13]{AA}
Andreas Andersson.
\newblock Operator deformations in quantum measurement theory.
\newblock {\em Letters in Mathematical Physics}, pages 1--16, November 2013.

\bibitem[BEH08]{BEH}
J.~{Blank}, P.~{Exner}, and M.~{Havl{\'i}\v{c}ek}.
\newblock {\em {Hilbert Space Operators in Quantum Physics}}.
\newblock Springer, 2008.

\bibitem[BLS11]{BLS}
Detlev Buchholz, Gandalf Lechner, and Stephen~J. Summers.
\newblock {Warped Convolutions, Rieffel Deformations and the Construction of
  Quantum Field Theories}.
\newblock {\em Commun.Math.Phys.}, 304:95--123, 2011.

\bibitem[BR96]{BR}
O.~Bratteli and D.W. Robinson.
\newblock {Operator algebras and quantum statistical mechanics. Vol. 2:
  Equilibrium states. Models in quantum statistical mechanics}.
\newblock 1996.

\bibitem[BS]{BS}
Detlev Buchholz and Stephen~J. Summers.
\newblock {Warped Convolutions: A Novel Tool in the Construction of Quantum
  Field Theories}.
\newblock {\em Quantum Field Theory and Beyond, pp. 107–121. World
  Scientific, Singapore.}

\bibitem[GL07]{GL1}
Harald Grosse and Gandalf Lechner.
\newblock {Wedge-Local Quantum Fields and Noncommutative Minkowski Space}.
\newblock {\em JHEP}, 0711:012, 2007.

\bibitem[GL08]{GL2}
Harald Grosse and Gandalf Lechner.
\newblock {Noncommutative Deformations of Wightman Quantum Field Theories}.
\newblock {\em JHEP}, 0809:131, 2008.

\bibitem[GS68]{GS1}
I.M. Gel'fand and G.E. Shilov.
\newblock {\em Generalized functions}.
\newblock Generalized Functions. Academic Press, 1968.

\bibitem[H{\"o}r04]{H}
L.~H{\"o}rmander.
\newblock {\em The Analysis of Linear Partial Differential Operators II:
  Differential Operators with Constant Coefficients}.
\newblock Springer, 2004.

\bibitem[Jos99]{Jo}
M.~S. Joshi.
\newblock {Introduction to pseudo-differential operators}.
\newblock 1999.

\bibitem[Lec12]{GL4}
Gandalf Lechner.
\newblock {Deformations of quantum field theories and integrable models}.
\newblock {\em Commun.Math.Phys.}, 312:265--302, 2012.

\bibitem[LST13]{GL5}
Gandalf Lechner, Jan Schlemmer, and Yoh Tanimoto.
\newblock {On the equivalence of two deformation schemes in quantum field
  theory}.
\newblock {\em Lett.Math.Phys.}, 103:421--437, 2013.

\bibitem[LW11]{LW}
Gandalf Lechner and Stefan Waldmann.
\newblock {Strict deformation quantization of locally convex algebras and
  modules}.
\newblock 2011.

\bibitem[MM11]{Mor}
Eric Morfa-Morales.
\newblock {Deformations of quantum field theories on de Sitter spacetime}.
\newblock {\em J.Math.Phys.}, 52:102304, 2011.

\bibitem[{Muc}12]{MUc}
A.~{Much}.
\newblock {Wedge-local quantum fields on a nonconstant noncommutative
  spacetime}.
\newblock {\em Journal of Mathematical Physics}, 53(8):082303, August 2012.

\bibitem[{Muc}14]{Muc1}
A.~{Much}.
\newblock {Quantum Mechanical Effects from Deformation Theory}.
\newblock {\em Journal of Mathematical Physics}, 55(8):022302, February 2014.

\bibitem[{Muc}15a]{Muc4}
A.~{Much}.
\newblock {Isospectral Deformations: The Massive Case}.
\newblock 2015.

\bibitem[{Muc}15b]{Muc3}
A.~{Much}.
\newblock {Relativistic corrections to the Moyal-Weyl spacetime}.
\newblock {\em Journal of Mathematical Physics}, 56(2):022301, February 2015.

\bibitem[Rie93]{RI}
M.A. Rieffel.
\newblock {Deformation quantization for actions of $\mathbb{R}^{d}$}.
\newblock {\em Memoirs A.M.S.}, 506, 1993.

\bibitem[RS75a]{RS1}
M.~Reed and B.~Simon.
\newblock {\em {Methods of Modern Mathematical Physics. 1. Functional
  Analysis}}.
\newblock {Gulf Professional Publishing}, 1975.

\bibitem[RS75b]{RS}
M.~Reed and B.~Simon.
\newblock {\em {Methods of Modern Mathematical Physics. 2. Fourier Analysis,
  Selfadjointness}}.
\newblock {Gulf Professional Publishing}, 1975.

\bibitem[Tes01]{T}
G.~Teschl.
\newblock {\em {Mathematical Methods in Quantum Mechanics}}, volume~99.
\newblock 2001.

\bibitem[Thi81]{Th}
W.E. Thirring.
\newblock {\em Quantum mechanics of atoms and molecules}.
\newblock A Course in Mathematical Physics. Springer-Verlag, 1981.

\end{thebibliography}

 \end{document}